\documentclass{article}
\pdfoutput=1

\usepackage{amstext,amsmath,amssymb,amsthm}
\usepackage{hyperref}

\newcommand{\defeq}{\mathrel{:=}}
\newcommand{\states}{\mathcal{S}}
\newcommand{\B}[2]{\mathsf{B}_{#1}^{#2}}
\newcommand{\K}[2]{\mathsf{K}_{#1}^{#2}}

\newtheorem{definition}{Definition}
\newtheorem{fact}{Fact}
\newtheorem{proposition}{Proposition}
\newtheorem{theorem}{Theorem}

\newenvironment{myproof}{\begin{proof}}{\end{proof}}%{\qed\end{proof}}

\begin{document}
\title{Parametric Constructive Kripke-Semantics for Standard Multi-Agent Belief and Knowledge\\
		(Knowledge As Unbiased Belief)}
%\author{Simon Kramer\inst{1} \and Joshua Sack\inst{2}}
%\institute{Centre for Security, Reliability and Trust\\
%			University of Luxembourg\\ 
%			\email{simon.kramer@a3.epfl.ch} \and 
%			Department of Mathematics and Statistics\\
%			California State University Long Beach\\
%			\email{joshua.sack@gmail.com}}
\author{%
	\begin{tabular}{@{}c@{}}
		Simon Kramer\thanks{funded with 
			Grant AFR~894328 from the National Research Fund Luxembourg  
				cofunded under the Marie-Curie Actions of the European Commission (FP7-COFUND)}\\[\jot] 
		%\normalsize Centre for Security, Reliability and Trust\\
		\normalsize University of Luxembourg\\
		\normalsize \texttt{simon.kramer@a3.epfl.ch}
	\end{tabular} 
	\and 
	\begin{tabular}{@{}c@{}}	
		Joshua Sack\\[\jot]
		%\normalsize Department of Mathematics and Statistics\\
		\normalsize California State University Long Beach\\
		%\normalsize University of Amsterdam (ILLC)\\
		\normalsize \texttt{joshua.sack@gmail.com}
	\end{tabular}}
\maketitle

\begin{abstract}
	We propose 
		\emph{parametric constructive Kripke-semantics} for 
			multi-agent KD45-belief and S5-knowledge in terms of 
				elementary set-theoretic constructions of 
					two basic functional building blocks, namely 
						\emph{bias} (or \emph{viewpoint}) and \emph{visibility,} functioning also as 
							the parameters of the doxastic and epistemic accessibility relation.
	The doxastic accessibility relates two possible worlds whenever 
		the application of the 
			composition of bias with visibility to the first world is equal to
		the application of visibility to the second world.
	The epistemic accessibility is 
		the transitive closure of 
			the union of 
				our doxastic accessibility and 
				its converse.
	Therefrom, accessibility relations for 
		\emph{common} and \emph{distributed} belief and knowledge can be constructed in a standard way.
	As a result, 
		we obtain \emph{a general definition of knowledge in terms of belief} that  
			enables us to view 
				S5-knowledge as accurate (unbiased and thus true) KD45-belief, 
				negation-complete belief and knowledge as 
					exact KD45-belief and S5-knowledge, respectively, and 
				perfect S5-knowledge as precise (exact and accurate) KD45-belief, and all this 
					\emph{generically} for arbitrary functions of bias and visibility.
	Our results can be seen as a semantic complement to 
		previous foundational results by Halpern et al.\ about 
			the (un)definability and (non-)reducibility of knowledge in terms of and to belief, respectively.

%	\keywords{parametric constructive Kripke-semantics,  
%				doxastic \& epistemic logic,
%				knowledge as a form of belief,
%				multi-agent distributed systems,
%				relational semantics of modal logic.}
	
	\medskip
	
\end{abstract}

\section{Introduction}
In \cite{KnowledgeAsBelief},
	the problem of defining knowledge in terms of belief is studied from
		a modal logic perspective, where
	the authors show that
		``if knowledge satisfies any set of axioms contained in S5, 
			then it cannot be explicitly defined in terms of belief. 
			S5 knowledge can be implicitly defined by belief, but not reduced to it.''
Thereby, 
	the standard notions of explicit and implicit definability from first-order logic are 
		``lifted to the definability of modalities in modal logics in a straightforward way,''
	so that 
	``explicit definability is equivalent to the combination of implicit definability and reducibility.''
More precisely, \cite{KnowledgeAsBelief}:
\begin{quotation}
	Consider a logic $\Lambda$ for knowledge and belief. 
	Knowledge is explicitly defined in $\Lambda$ if 
	there is a formula $\mathrm{DK}$ (for ``definition of knowledge'') in $\Lambda$ of the form 
	$Kp\leftrightarrow\delta$, where $\delta$ is a formula that does not mention the knowledge operator.
	Knowledge is implicitly defined in $\Lambda$ if, roughly speaking, $\Lambda$ ``determines'' knowledge
	uniquely. Syntactically, this determination means that any two modal operators for knowledge that 
	satisfy $\Lambda$ must be equivalent. Semantically, this means that two Kripke models of $\Lambda$ 
	with the same set of worlds that agree on the interpretation of belief (and on the interpretations 
	of all primitive propositions) must agree also on the interpretation of knowledge.
\end{quotation}
Our contribution is to make 
	the definability of S5-knowledge in terms of 
		KD45-belief \emph{function-parametric} as well as \emph{semantic-constructive} 
			(cf.\ Definition~\ref{definition:FunctionParametricDoxasticAccessibility}--\ref{definition:FPDEL}).
More precisely,
	we propose 
		function-parametric constructive Kripke-semantics for 
			multi-agent KD45-belief and S5-knowledge in terms of 
				elementary set-theoretic constructions of 
					two basic functional building blocks, namely (cf.\ Definition~\ref{definition:DEFP}):  
					\begin{itemize}
						\item \emph{bias,} or viewpoint translocation (necessarily idempotent, e.g., 
								the constant functions), and
						\item \emph{visibility} transformation, for example: 
							\begin{itemize}
								\item point confounding (non-injective when non-trivial), and/or 
								\item point confusing (or permuting, bijective on a sub-domain)
							\end{itemize}
					\end{itemize}
			functioning also as the parameters of the doxastic and epistemic accessibility relation.
	Note that we mean ``\emph{set-theoretically} constructive,'' not ``intuitionistic,''
		in loose analogy with 
							the set-theoretically constructive rather than 
							the purely axiomatic definition of 
								numbers \cite{TheNumberSystems} or 
								ordered pairs.\footnote{E.g., the now standard definition by Kuratowski 
									or other well-known definitions \cite{NotesOnSetTheory}.}
	That is:
			\begin{itemize}
				\item our \emph{epistemic} accessibility is   
				the transitive closure of 
					the union of 
						our doxastic accessibility and 
						its converse  
							(cf.\ Definition~\ref{definition:FunctionParametricEpistemicAccessibility}); 
				\item our \emph{doxastic} accessibility relates two possible worlds whenever 
						the application of the 
							composition of bias with visibility to the first world is equal to
							the application of visibility to the second world
								(cf.\ Definition~\ref{definition:FunctionParametricDoxasticAccessibility}).
			\end{itemize}
	As a result, our constructions enable us to view 
		\begin{itemize}
			\item S5-knowledge as accurate (unbiased and thus true) KD45-belief, 
			\item \emph{negation-complete} belief and knowledge as 
					exact KD45-belief and S5-knowledge, respectively, 
			\item \emph{perfect} S5-knowledge as precise (exact and accurate) KD45-belief, and
		\end{itemize}
	all this 
					\emph{generically} for arbitrary functions of bias and visibility in our sense 
						(cf.\ Theorem~\ref{theorem:DoxasticEpistemic}).
	In comparison, 
		recall from \cite{Epistemic_Logic} the by-now classic constructive definition of 
			agent-centric (say in agent $a$) epistemic accessibility $\equiv_{a}$ as state (say $s$ and $s'$) indistinguishability $s\equiv_{a}s'$ defined in terms of 
				the equality between  
					the projection $\pi_{a}(s)$ of $s$ onto $a$'s view and 
					the projection $\pi_{a}(s'):$
	\begin{definition}[Epistemic accessibility as state indistinguishability \cite{Epistemic_Logic}]\label{definition:EAasI}
		$$\text{$s\equiv_{a}s'$ by definition, if and only if $\pi_{a}(s)=\pi_{a}(s')$.}$$
	\end{definition}
	Thus $\equiv_{a}$ is defined to be the kernel of $\pi_{a}$ \cite{DataMining}.
	This definition is constructive in the sense that
		it not merely abstractly stipulates $\equiv_{a}$
			to be an equivalence relation (that would be the standard modal-logical methodology
				\cite{MultiAgents}) but 
				it actually concretely constructs $\equiv_{a}$ in terms of 
					the set-theoretic building block $\pi_{a}$, a projection function
						(state visibility as state projection), which
						\emph{forces} $\equiv_{a}$ to be an equivalence relation.
	(For more examples of more complex, constructive definitions of agent-centric accessibility relations, 
		see \cite{LiP,LiiP,LDiiP}.)
	It is not at all obvious how to recover a definition of doxastic accessibility from this 
		indistinguishability definition of epistemic accessibility.
Nevertheless we present a simple, generic, and thus general solution to this important problem.
Our solution is general in the sense that 
	the extent to which it can be applied is 
		the entire semantic scope (models) of standard doxastic and epistemic logic 
			(cf.\ \cite{MultiAgents} and \cite{EpistemicLogicFiveQuestions} for overviews), 
				thanks to our soundness and completeness results 
					in the sense of Theorem~\ref{theorem:KD45} and \ref{theorem:S5}.
Moreover, our proofs for the solution are simple, which increases its value.
Here, the difficulty was to \emph{find} our general definition of knowledge in terms of belief, which
	has even the feature of being generic thanks to its function parameters.
Our findings can be seen as a semantic complement to 
		previous foundational results by Halpern et al.\ about 
			the (un)definability and (non-)reducibility of knowledge in terms of and to belief, respectively.

\section{Parametric constructions and results}
Let 
	\begin{itemize}
	\item $\states$ designate a set of 
		system states in computer-science, 
		points in modal-logical, or  
		possible worlds in philosophical terminology;
	\item $\mathrm{id}_{\states'}\defeq\{(s,s)\mid s\in\states'\}$ the identity function on $\states'\subseteq\states\,;$ 
	\item $\mathrm{Im}(R)\defeq\{s'\in\states\mid\text{there is $s\in\states$ such that $s\mathrel{R}s'$}\}$ the image of 
	some (possibly functional) relation $R\subseteq\states\times\states\,.$ 
	%\item $R(s)\defeq\{s'\in\states\mid s\mathrel{R}s'\}$ the image of $s$ under $R$.
\end{itemize}
Further let ``:iff'' abbreviate ``by definition, if and only if''.
\begin{definition}[Doxastic-epistemic function pair]\label{definition:DEFP}
Two functions 
	$f:\states\rightarrow\states$ and 
	$g:\mathrm{Im}(f)\rightarrow\mathrm{Im}(f)$ form  
	a \emph{doxastic-epistemic function pair $(f,g)$ on $\states$} :iff 
\begin{itemize}
	\item for all $s,s'\in\states$, if $g(f(s))=f(s')$ then $g(f(s'))=f(s')\,;$
	\item \textbf{or, equivalently,} $g$ is idempotent, i.e., $g\circ g=g\,.$ 
\end{itemize}
\end{definition}
\begin{fact}\label{fact:DEFPCOne}
	For all 
		doxastic-epistemic function pairs $(f,g)$ on $\states$ and 
		$s\in\states$ there is $s'\in\states$ such that $g(f(s))=f(s')\,.$
\end{fact}
\begin{myproof}
	By the definitional fact that 
		$f$ is a totally defined operation on $\states$ and 
		$g$ is a totally defined operation on $\mathrm{Im}(f)$;  
		($\mathrm{Im}(g)\subseteq\mathrm{Im}(f)$).
\end{myproof}
We shall use 
	the two constraints in Definition~\ref{definition:DEFP} interchangeably; 
the two constraints are indeed equivalent, as asserts the following proposition.
\begin{proposition}
	The two alternative constraints in Definition~\ref{definition:DEFP} are equivalent.
\end{proposition}
\begin{myproof}
	Let 
		$f:\states\rightarrow\states$ and 
		$g:\mathrm{Im}(f)\rightarrow\mathrm{Im}(f)$.
	For the if-direction, 
		suppose that for all $s,s'\in\states$, 
			if $g(f(s))=f(s')$ then $g(f(s'))=f(s')$.
	Further let $s\in\states$.
	By Fact~\ref{fact:DEFPCOne}, 
		there is $s'\in\states$ such that $g(f(s))=f(s')$.
	Hence $g(f(s))=g(f(s'))$, and also 
		$g(g(f(s)))=g(f(s'))$.
	Hence $g(g(f(s)))=g(f(s))$.
	For the only-if-direction,
		suppose that $g\circ g=g$, and 
		let $s,s'\in\states$.
	Further suppose that $g(f(s))=f(s')$.
	Hence $g(g(f(s)))=g(f(s'))$.
	Hence $g(f(s))=g(f(s'))$ by the idempotency of $g$.
	Hence $g(f(s'))=f(s')$ by the last supposition.
\end{myproof}

\begin{definition}[Function-Parametric \emph{Doxastic} Accessibility]\label{definition:FunctionParametricDoxasticAccessibility}
Let $(f,g)$ designate a doxastic-epistemic function pair on $\states$.
Then we define 
	our \emph{$(f,g)$-parametric \textbf{doxastic} accessibility relation} ${\mathrm{D}_{f}^{g}}\subseteq\states\times\states$ such that 
for all $s,s'\in\states$,
	$$\boxed{\text{$s\mathrel{\mathrm{D}_{f}^{g}}s'$ :iff\ \ $g(f(s))=f(s')\,.$}}$$
\end{definition}
The following main \emph{adequacy theorem} asserts 
	first that for all doxa\-stic-epistemic function pairs $(f,g)$, 
		$\mathrm{D}_{f}^{g}$ is indeed a standard doxastic accessibility relation, and 
	second that for all standard doxastic accessibility relations $R$, 
		a doxastic-epistemic function pair $(f,g)$ can be constructed such that 
			$R=\mathrm{D}_{f}^{g}$.
\begin{theorem}[The \textbf{KD45} Accessibility Schema]\label{theorem:KD45}\ 
\begin{enumerate}
	\item \textbf{Soundness:} If $(f,g)$ is a doxastic-epistemic function pair on $\states$
								then for all $s\in\states:$
	\begin{enumerate}
		\item there is $s'\in\states$ such that $s\mathrel{\mathrm{D}_{f}^{g}}s'$\quad(Seriality/Totality)
		\item for all $s',s''\in\states:$ 
			\begin{enumerate}
				\item if $s\mathrel{\mathrm{D}_{f}^{g}}s'$ and $s'\mathrel{\mathrm{D}_{f}^{g}}s''$
						then $s\mathrel{\mathrm{D}_{f}^{g}}s''$\quad(Transitivity)
				\item if $s\mathrel{\mathrm{D}_{f}^{g}}s'$ and $s\mathrel{\mathrm{D}_{f}^{g}}s''$
						then $s'\mathrel{\mathrm{D}_{f}^{g}}s''$\quad(Euclideanness)
			\end{enumerate}
	\end{enumerate}
	\item \textbf{Completeness:} If $\emptyset\neq R\subseteq\states\times\states$ is serial, transitive, and Euclidean  
		then there is a doxastic-epistemic function pair $(f,g)$ on $\states$ such that 
			$(f,g)$ is constructible from $R$ and 
			$R=\mathrm{D}_f^g\,.$
\end{enumerate}
\end{theorem}
\begin{myproof}
	For soundness, assume that 
		$(f,g)$ is a doxastic-epistemic %\linebreak 
		function pair on $\states$.
	Then 1.a holds by Fact~\ref{fact:DEFPCOne}.
	For 1.b, 
		let $s,s',s''\in\states$ and 
		suppose that $g(f(s))=f(s')$.
	Hence $g(f(s'))=f(s')$ by the first alternative definitional constraint on $f$ and $g$.
	For 1.b.i, 
		further derive that $g(f(s))=g(f(s'))$ by 
			the last supposition, and then 
		further suppose that $g(f(s'))=f(s'')$.
	Hence $g(f(s))=f(s'')$. 
	For 1.b.ii, 
		suppose that $g(f(s))=f(s'')$.
	Consequently, 
		$f(s')=f(s'')$ by the last supposition and the first supposition of 1.b, and then 
		$g(f(s'))=f(s'')$ by the very first derivation.
		
	For completeness,
		let $\emptyset\neq R\subseteq\states\times\states$ and 
		suppose that $R$ is serial, transitive, and Euclidean.
Then ${\equiv}\defeq\{(s,s')\in R\mid s\in\mathrm{Im}(R)\}\subseteq R$ is an equivalence relation 
(i.e., a relation that is reflexive, transitive, and Euclidean) on $\mathrm{Im}(R):$
\begin{itemize}
	\item $\equiv$ is reflexive: 
		Let $s\in\mathrm{Im}(R)$, i.e.,
		there is $s'\in\states$ such that $s'\mathrel{R}s$.
		Thus $s\in\states$.
		Hence there is $s''\in\states$ such that $s\mathrel{R}s''$ by the seriality of $R$.
		Hence $s'\mathrel{R}s''$ by the transitivity of $R$.
		Hence $s''\mathrel{R}s$ by the Euclideanness of $R$.
		Hence $s\mathrel{R}s$ by the transitivity of $R$.
		And since $s\in\mathrm{Im}(R)$, $s\equiv s$.
	\item $\equiv$ is transitive and Euclidean by inheritance, i.e., simply because $R$ is.
\end{itemize}
For each $s\in \mathrm{Im}(R)$, 
	choose $c_s\in[s]_{\equiv}$ such that 
		for all $s'\in\states$, 
			if $[s]_\equiv = [s']_\equiv$ then $c_s=c_{s'}.$
Observe that 
	for all $s',s''\in\states$,
		if $s\mathrel{R}s'$ and $s\mathrel{R}s''$
		then there is $C\in\mathrm{Im}(R)/_{\equiv}$ such that 
			$s',s''\in C$.
This is 
	because $s'\mathrel{R}s''$ by the Euclideanness of $R$ and 
	because $s',s''\in\mathrm{Im}(R)$ (thus $s'\equiv s''$).
Now define two functions $f:\states\to \states$ and $g:\mathrm{Im}(f)\to\mathrm{Im}(f)$ such that:
	\begin{align*}
		f:s&\mapsto \left\{\begin{array}{ll}
						c_s & \mbox{if $s\in \mathrm{Im}(R)\,,$}\\
						s & \mbox{if $s\not\in \mathrm{Im}(R)\,;$}
					\end{array}\right . \\
		g:s&\mapsto \left\{\begin{array}{ll}
						s & \mbox{if $s\in \mathrm{Im}(R)$\,,}\\
						c_{s'} & \mbox{if $s\not\in \mathrm{Im}(R)$ and $s\mathrel{R}s'$.}
					\end{array}\right .
	\end{align*}
Notice that 
	$g$ is well defined, i.e., it does not matter which $s'$ we choose, since
		for all $s''\in\states$,
			if $s\mathrel{R}s'$ and $s\mathrel{R}s''$ 
			then there is $C\in\mathrm{Im}(R)/_{\equiv}$ such that $s',s''\in C$.

We will now see that $R = \mathrm{D}_f^g$.
So let $s,s'\in\states$.
\begin{itemize}
	\item Suppose that $s\mathrel{R}s'$. Thus $s'\in\mathrm{Im}(R)$. Hence $f(s') = c_{s'}$.
		\begin{itemize}
			\item Suppose that $s\in \mathrm{Im}(R)$. 
					Hence  
						$s\equiv s'$ by the first two suppositions (thus $c_s = c_{s'}$), and 
						$g(f(s)) = g(c_s) = c_s$.
					Hence $g(f(s))=f(s')$.
			\item Now suppose that $s\not\in \mathrm{Im}(R)$.
					Hence $g(f(s)) = g(s) = c_{s''}$ for some $s''\in\states$ such that $s\mathrel{R}s''$.
					Hence $s'\mathrel{R}s''$ by the Euclideanness of $R$, and
					then $s'\equiv s''$.
					Thus $c_{s'}=c_{s''}$.
					Hence $g(f(s)) = f(s')$.
		\end{itemize}
	\item Conversely suppose that $g(f(s)) = f(s')$.
			Notice in the definition of $f$ and $g$ 
				that $\mathrm{Im}(g)\subseteq\mathrm{Im}(R)$ and 
				that $s'\in \mathrm{Im}(R)$ if and only if 
						$f(s')\in \mathrm{Im}(R)$. 
				Hence $s'\in \mathrm{Im}(R)$, and thus $f(s') = c_{s'}$.
				Hence $g(f(s)) = c_{s'}$.
		\begin{itemize}
			\item Suppose that $s\in \mathrm{Im}(R)$.
					Hence $g(f(s)) = g(c_s) = c_s$.
					Hence $c_s= c_{s'}$ and
					thus $s\equiv s'$.
					Hence $s\mathrel{R}s'$.
			\item Now suppose that $s\not\in \mathrm{Im}(R)$.
					Hence 
						$g(f(s)) = g(s) = c_{s''}$ for some $s''\in\states$ such that 
						$s\mathrel{R}s''$.
					Hence $c_{s''}= c_{s'}$, and
					thus $s''\equiv s'$.
					Hence $s''\mathrel{R}s'$.
					Hence $s\mathrel{R}s'$ by the transitivity of $R$.
		\end{itemize}
\end{itemize}
\end{myproof}
The following proposition gives 
	a functional characterisation of (i.e., a necessary and sufficient equational condition for)  
		the symmetry (and hence the property of being an equivalence relation) of $(f,g)$-parametric doxastic accessibility relations.
(Seriality, symmetry, and transitivity jointly imply reflexivity.) 
\begin{proposition}[Doxastic symmetry characterisation]\label{proposition:DoxasticSymmetry}
	For all $(f,g)$-parametric doxastic accessibility relations 
		$\mathrm{D}_{f}^{g}\subseteq\states\times\states$, 
	$$\text{$\mathrm{D}_{f}^{g}=(\mathrm{D}_{f}^{g})^{-1}$ if and only if 
		$g=\mathrm{id}_{\mathrm{Im}(f)}\,.$}$$
\end{proposition}
\begin{proof}
	The if-direction is immediate.
	For the only-if-direction,
		suppose that $\mathrm{D}_{f}^{g}=(\mathrm{D}_{f}^{g})^{-1}$, i.e., 
			for all $s,s'\in\states$,
				$g(f(s))=f(s')$ if and only if $g(f(s'))=f(s)$.
	Further, let $s\in\states$.
	Hence there is $s'\in\states$ such that $g(f(s))=f(s')$ by Fact~\ref{fact:DEFPCOne}.
	Hence 
		$g(f(s'))=f(s)$ by the first supposition, and also
		$g(g(f(s)))=g(f(s'))$. 
	Hence, 
		$g(f(s))=g(f(s'))$ by the idempotency of $g$, and then
		$g(f(s))=f(s)$.
\end{proof}
The following is our general definition of knowledge in terms of belief.
\begin{definition}[Function-Parametric \emph{Epistemic} Accessibility]
\label{definition:FunctionParametricEpistemicAccessibility}
Let $\mathrm{D}_{f}^{g}$ designate an $(f,g)$-parametric \emph{doxastic} accessibility. 
Then we define our \emph{$(f,g)$-parametric \textbf{epistemic} accessibility relation} ${\mathrm{E}_{f}^{g}}\subseteq\states\times\states$ such that 
		$$\boxed{{\mathrel{\mathrm{E}_{f}^{g}}}\defeq  
			({\mathrel{\mathrm{D}_{f}^{g}}}\cup 
			(\mathrel{\mathrm{D}_{f}^{g}})^{-1})^{+},}$$
	where 
		`$^{-1}$' designates the converse and 
		`$^{+}$' the transitive-closure operation.
\end{definition}
The following \emph{adequacy theorem} asserts 
	first that for all doxastic-epistemic function pairs $(f,g)$, 
		$\mathrel{\mathrm{E}_{f}^{g}}$ is indeed a standard epistemic accessibility relation, and 
	second that for all standard epistemic accessibility relations $\equiv$, 
		a doxastic-epistemic function pair $(f,g)$ can be constructed such that ${\equiv}=\mathrm{E}_{f}^{g}$.
\begin{theorem}[The \textbf{S5} Accessibility Schema]\label{theorem:S5}\ 
\begin{enumerate}
	\item \textbf{Soundness:} If $\mathrm{D}_{f}^{g}$ is an $(f,g)$-parametric accessibility relation 
		then $\mathrm{E}_{f}^{g}$ is the smallest equivalence relation containing $\mathrm{D}_{f}^{g}$.
	\item \textbf{Completeness:} If $\emptyset\neq{\equiv}\subseteq\states\times\states$ is an equivalence relation
		then there is a doxastic-epistemic function pair $(f,g)$ on $\states$ such that 
			$(f,g)$ is constructible from $\equiv$ and 
			${\equiv}=\mathrm{E}_f^g\,.$ 
\end{enumerate}
\end{theorem}
\begin{myproof}
For soundness, consider that 
	since $\mathrm{D}_{f}^{g}$ is serial and \emph{transitive} by construction, and
	since the converse operation preserves the seriality and transitivity of $\mathrm{D}_{f}^{g}$,
		$\mathrm{E}_{f}^{g}$ is so too.
Finally, 
	since $\mathrm{E}_{f}^{g}$ is \emph{symmetric} by construction,
	$\mathrm{E}_{f}^{g}$ is also \emph{reflexive,} and thus an equivalence relation 
		containing $\mathrm{D}_{f}^{g}$.
(Seriality, symmetry, and transitivity jointly imply reflexivity.)
To see that 
	$\mathrm{E}_{f}^{g}$ is the \emph{smallest} such relation,
		recall from \cite{DataMining} that 
			for arbitrary $R\subseteq\states\times\states$,
				the relation $(R\cup(R)^{-1}\cup\mathrm{id}_{\states})^{*}$ is
					the smallest equivalence relation containing $R$, where
						`$^{*}$' is the reflexive-transitive-closure operation.
However,
	we can spare $\mathrm{id}_{\states}$ and the reflexive closure,
		since $\mathrm{id}_{\states}\subseteq\mathrm{E}_{f}^{g}$ is reflexive by construction.

For completeness,
	suppose that $\emptyset\neq{\equiv}\subseteq\states\times\states$ is an equivalence relation.
Then for each equivalence class $C\in\states/_{\equiv}$, choose $c\in C$ and 
define $\pi(s)\defeq c$	for all $s\in C$.
Clearly, 
	${\equiv}=\mathrm{D}_{\pi}^{\mathrm{id}_{\mathrm{Im}(\pi)}}$, and
	$\mathrm{D}_{\pi}^{\mathrm{id}_{\mathrm{Im}(\pi)}}=\mathrm{E}_{\pi}^{\mathrm{id}_{\mathrm{Im}(\pi)}}$
		(see also Proposition~\ref{proposition:EDD}.2).
Thus 
	${\equiv}=\mathrm{E}_{\pi}^{\mathrm{id}_{\mathrm{Im}(\pi)}}$.
\end{myproof}

%\begin{proposition}[Doxastic-epistemic characterisation]\label{proposition:DEC}
%	For all $(f,g)$-parametric doxastic and epistemic accessibility relations
%	$\mathrm{D}_{f}^{g}\subseteq\states\times\states$ and $\mathrm{E}_{f}^{g}\subseteq\states\times\states$, respectively, and 
%	all $s,s'\in\states$,
%		$$\text{$s\mathrel{\mathrm{E}_{f}^{g}}s'$ if and only if
%		$\mathrm{D}_{f}^{g}(s)=\mathrm{D}_{f}^{g}(s')$\,.}$$
%\end{proposition}
%\begin{proof}
%	\ldots
%\end{proof}

The following proposition is the basis for our third main result, namely 
Theorem~\ref{theorem:DoxasticEpistemic}.
\begin{proposition}[Doxastic-epistemic accessibility inclusions]\label{proposition:EDD}
For all doxastic-epistemic function pairs $(f,g)$ on $\states:$
\begin{enumerate}
	\item $%{\mathrm{E}_{f}^{g}}\circ{\mathrm{D}_{f}^{g}}\subseteq
	{\mathrm{D}_{f}^{g}}\subseteq{\mathrm{E}_{f}^{g}}\,;$
	\item ${\mathrm{D}_{f}^{g}}={\mathrm{E}_{f}^{g}}$ if and only if
			$g=\mathrm{id}_{\mathrm{Im}(f)}\,;$
	\item $\mathrm{id}_{\states}=
			{\mathrm{E}_{\mathrm{id}_{\states}}^{\mathrm{id}_{\states}}}=
			{\mathrm{D}_{\mathrm{id}_{\states}}^{\mathrm{id}_{\states}}}\,.$
\end{enumerate}
\end{proposition}
\begin{myproof}
	For %the right inclusion in 
	1, inspect definitions.
	%
	%For the left inclusion in 1, 
	%	let $s,s'\in\states$, and 
	%	suppose that $s\mathrel{\mathrm{E}_{f}^{g}\circ\mathrm{D}_{f}^{g}}s'$.
	%
	%Hence there is $s''\in\states$ such that
	%	$s\mathrel{\mathrm{E}_{f}^{g}}s''$ and $s''\mathrel{\mathrm{D}_{f}^{g}}s'$.
	%Hence (for all $s'''\in\states$,
	%	$s\mathrel{\mathrm{D}_{f}^{g}}s'''$ if and only if 
	%	$s''\mathrel{\mathrm{D}_{f}^{g}}s'''$) and $s''\mathrel{\mathrm{D}_{f}^{g}}s'$ 
	%		by Proposition~\ref{proposition:DEC}.
	%
	%Hence 
	%	$s\mathrel{\mathrm{D}_{f}^{g}}s'$ if and only if 
	%	$s''\mathrel{\mathrm{D}_{f}^{g}}s'$) and $s''\mathrel{\mathrm{D}_{f}^{g}}s'$.
	%
	%Hence $s\mathrel{\mathrm{D}_{f}^{g}}s'$.
	%
	2 follows from 
			Proposition~\ref{proposition:DoxasticSymmetry} and
			the definition of ${\mathrm{E}_{f}^{g}}\,.$
	For 3, inspect 2.
\end{myproof}	

\begin{definition}[Doxastic-epistemic similarity type] \label{definition:DEST}
Let 
\begin{itemize}
	\item $\mathcal{P}\neq\emptyset$ designate some set of \emph{atomic propositions} $P\,;$ 
	\item $\mathcal{T}$ a set of \emph{types} $T$ such that $\mathtt{S}\in\mathcal{T}\,;$  
	\item $\mathcal{G}$ a set of typed function names $\texttt{g}:T\to T'$ 
			(abbreviated as $\texttt{g}$ when clear from context) such that 
				$(\mathtt{id}_T:T\to T)\in\mathcal{G}$ for all $T\in \mathcal{T}\,;$
	\item $\mathcal{F}$ a set of typed function names $\texttt{f}:\mathtt{S}\to T$ 
			(abbreviated as $\texttt{f}$ when clear from context) such that
				$(\mathtt{id}_\mathtt{S}:\mathtt{S}\to\mathtt{S})\in\mathcal{F}$ and 
				if 
					$(\texttt{f}:\mathtt{S}\to T)\in\mathcal{F}$ and 
					$(\texttt{h}:T\to T')\in\mathcal{F}\cup\mathcal{G}$
				then $(\texttt{h}\circ_{T}\texttt{f}:\mathtt{S}\to T')\in\mathcal{F}\,;$
	\item $\Pi_B, \Pi_K\subseteq \Pi\defeq\{(\texttt{f}:T\to T', \texttt{g}:T'\to T'')\in \mathcal{F}\times \mathcal{G}\}$ \emph{belief-label} and \emph{knowledge-label sets,} and 
	the so-called \emph{label set}, respectively.
\end{itemize}
Then, $$\Sigma\defeq(\mathcal{P},\mathcal{T},\mathcal{G},\mathcal{F},\Pi_B, \Pi_K)$$
	is a \emph{doxastic-epistemic similarity type}.
\end{definition}
Note the above-introduced notational conventions:
	we use 
		$\emph{\texttt{f}}$, $\emph{\texttt{g}}$, and $\emph{\texttt{h}}$ as meta-variables for typed function names, and 
		$f$, $g$, and $h$ as meta-variables for typed functions;  
	$\mathtt{id}_{\mathtt{S}}$ is an example of a (typed) function name, and
	$\mathrm{id}_{\states}$ is an example of a (typed) function.

\begin{definition}[Functional doxastic-epistemic language]\label{definition:FPLang}
Given a doxastic-epi\-stemic similarity type $\Sigma$ with   
	set $\mathcal{P}$ of atomic propositions, and  
	belief- and knowledge-label set $\Pi_B$ and $\Pi_K$, respectively, 
\begin{eqnarray*}
\mathcal{L}(\Sigma) \ni\phi&::=& P\; (P\in \mathcal{P})\mid\neg\phi\mid\phi\wedge\phi\\
&& \mid\B{\texttt{f}}{\texttt{g}}(\phi)\;((\texttt{f},\texttt{g})\in \Pi_B)\\
&& \mid\K{\texttt{f}}{\texttt{g}}(\phi)\;((\texttt{f},\texttt{g})\in \Pi_K)
\end{eqnarray*}
is the \emph{doxastic-epistemic language over $\Sigma$}.
\end{definition}
We intend for the operators $\B{\texttt{\emph{f}}}{\texttt{\emph{g}}}$ and $\K{\texttt{\emph{f}}}{\texttt{\emph{g}}}$ to have flexible readings, though they generally relate to belief and knowledge, respectively.  We may associate every pair $(\texttt{\emph{f}},\texttt{\emph{g}})$ with 
an agent $a$ in some given set $\mathcal{A}$ of agents (as we do in Proposition~\ref{proposition:RelatedWork}). 
When doing so, we may read $\B{\texttt{\emph{f}}}{\texttt{\emph{g}}}$ as 
``agent $a$ believes that $\phi$'' and $\K{\texttt{\emph{f}}}{\texttt{\emph{g}}}$ as 
``agent $a$ knows that $\phi$,'' where $a$ is the agent associated with the pair 
$(\texttt{\emph{f}},\texttt{\emph{g}})$.
\begin{definition}[Functional doxastic-epistemic models] \label{definition:FPM}
Given a doxastic-epist\-emic similarity type 
$\Sigma = (\mathcal{P},\mathcal{T},\mathcal{G},\mathcal{F},\Pi_B, \Pi_K)$,
			let
			\begin{itemize}
			\item $\langle \states, \iota\rangle$ be a so-called
				\emph{$\Sigma$-instantiation structure on $\states$} with an interpretation function $\iota$ 
					for types $T\in\mathcal{T}$ and typed function names constrained such that: 
				\begin{itemize}
					\item $\iota(\mathtt{S}) = \states$ and $\iota(T) \subseteq \states\,;$
					\item $\boxed{\begin{array}[t]{@{}r@{\ \ }c@{\ \ }l@{}}
							\iota(\mathtt{id}_{T}:T\to T)&\defeq&\mathrm{id}_{\iota(T)}\,,\\[\jot]
							\iota(\texttt{h}\circ_{T'}\texttt{f}:T\to T'')&\defeq&
								\iota(\texttt{h}:T'\to T'')\circ
								\iota(\texttt{f}:T\to T')\,;
							\end{array}}$
					\item $\iota(\texttt{h}:T\to T')$ is a function $h:\iota(T)\to \iota(T')$ such that
							if $\texttt{h}\in \mathcal{G}$ then $h$ is idempotent;
				\end{itemize}
				\item $\langle \states,  \{\mathrm{D}_{\iota(\texttt{f})}^{\iota(\texttt{g})}\}_{(\texttt{f},\texttt{g})\in\Pi_B},\{\mathrm{E}_{\iota(\texttt{f})}^{\iota(\texttt{g})}\}_{(\texttt{f},\texttt{g})\in\Pi_K}\rangle$ a \emph{doxastic-epistemic $\Sigma$-frame on $\langle \states, \iota\rangle;$}
				\item $\mathcal{V}:\mathcal{P}\rightarrow2^{\states}$ a 					standard modal \emph{valuation function} \cite{ModalLogicSemanticPerspective}, 
						mapping each atomic proposition $P$ to 
							the set of states where $P$ is considered true.
			\end{itemize}
			Then, 							
				$$\mathfrak{S}\defeq \langle \states, \{\mathrm{D}_{\iota(\texttt{f})}^{\iota(\texttt{g})}\}_{(\texttt{f},\texttt{g})\in\Pi_B},\{\mathrm{E}_{\iota(\texttt{f})}^{\iota(\texttt{g})}\}_{(\texttt{f},\texttt{g})\in\Pi_K},\mathcal{V}\rangle$$ is 
				the \textbf{\emph{doxastic-epistemic $\Sigma$-model on $\langle \states, \iota\rangle$ and $\mathcal{V}$}}, and 
				$(\mathfrak{S},s)$ a \emph{\textbf{pointed} doxastic-epistemic $\Sigma$-model on 
	$\langle \states, \iota\rangle$ and $\mathcal{V}$} for any $s\in \states$.
\end{definition}

\begin{definition}[Functional doxastic-epistemic logic]\label{definition:FPDEL}
Given a doxastic-epistemic similarity type $\Sigma = (\mathcal{P},\mathcal{T},\mathcal{G},\mathcal{F},\Pi_B, \Pi_K)$, define
\begin{itemize}
	\item a standard \emph{satisfaction relation} $\models$ between 
		pointed doxastic-epi\-stemic $\Sigma$-models and 
		their languages $\mathcal{L}(\Sigma)$ as in 
		Table~\ref{table:DoxasticEpistemicSatisfaction};
	\item $\mathfrak{S}\models\phi$ :iff for all $s\in\states$, $(\mathfrak{S},s)\models\phi\,;$
	\item $\models\phi$ :iff 
			for all doxastic-epistemic $\Sigma$-models $\mathfrak{S}$, 
						$\mathfrak{S}\models\phi\,.$
\end{itemize}
\begin{table}[t]
\centering
\caption{Doxastic-epistemic satisfaction relation}
%\smallskip
$\boxed{\begin{array}{@{}rcl@{}}
	(\mathfrak{S},s) \models P 
		&\text{:iff}& s\in\mathcal{V}(P)\\
	(\mathfrak{S},s) \models \neg\phi 
		&\text{:iff}& \text{not $(\mathfrak{S},s)\models\phi$}\\
	(\mathfrak{S},s) \models \phi\land\phi' 
		&\text{:iff}& \text{$(\mathfrak{S},s)\models\phi$ and $(\mathfrak{S},s)\models\phi'$}\\
	(\mathfrak{S},s) \models \B{\emph{\texttt{f}}}{\emph{\texttt{g}}}(\phi) 
		&\text{:iff}& \text{for all $s'\in\states$, 
			if $s\mathrel{\mathrm{D}_{\iota(\emph{\texttt{f}})}^{\iota(\emph{\texttt{g}})}}s'$ then $(\mathfrak{S},s')\models\phi$}\\
	(\mathfrak{S},s) \models \K{\emph{\texttt{f}}}{\emph{\texttt{g}}}(\phi) 
		&\text{:iff}& \text{for all $s'\in\states$, 
			if $s\mathrel{\mathrm{E}_{\iota(\emph{\texttt{f}})}^{\iota(\emph{\texttt{g}})}}s'$ then $(\mathfrak{S},s')\models\phi$}
\end{array}}$
\label{table:DoxasticEpistemicSatisfaction}
\end{table}
\end{definition}

\begin{proposition}[KD45-belief modality]\label{theorem:KD45BeliefModality}
For all $(\texttt{f},\texttt{g})\in \Pi_B$,
$\B{\texttt{f}}{\texttt{g}}$ is a \emph{KD45-belief modality.} That is:
	\begin{enumerate}
		\item $\models\B{\texttt{f}}{\texttt{g}}(\phi\rightarrow\phi')\rightarrow
						(\B{\texttt{f}}{\texttt{g}}(\phi)\rightarrow\B{\texttt{f}}{\texttt{g}}(\phi'))$\quad(Kripke's law, K)
		\item $\models\neg\B{\texttt{f}}{\texttt{g}}(\bot)$ 
				(equivalently, $\models\B{\texttt{f}}{\texttt{g}}(\phi)\rightarrow\neg\B{\texttt{f}}{\texttt{g}}(\neg\phi)$)\quad(belief consistency, D)
		\item $\models\B{\texttt{f}}{\texttt{g}}(\phi)\rightarrow\B{\texttt{f}}{\texttt{g}}(\B{\texttt{f}}{\texttt{g}}(\phi))$\quad(positive introspection, \textbf{4})
		\item $\models\neg\B{\texttt{f}}{\texttt{g}}(\phi)\rightarrow\B{\texttt{f}}{\texttt{g}}(\neg\B{\texttt{f}}{\texttt{g}}(\phi))$\quad(negative introspection, \textbf{5})
		\item if $\models\phi$ then $\models\B{\texttt{f}}{\texttt{g}}(\phi)$\quad(necessitation, N)
	\end{enumerate}
\end{proposition}
\begin{myproof}
	By Theorem~\ref{theorem:KD45}.1.
	K and N are forced by Kripke-semantics.
	The D-law corresponds to seriality,
		the 4-law to transitivity, and
		the 5-law to Euclideanness.
\end{myproof}

\begin{proposition}[S5-knowledge modality]\label{theorem:S5KnowledgeModality}
For all $(\texttt{f},\texttt{g})\in \Pi_K$,
$\K{\texttt{f}}{\texttt{g}}$ is an \emph{S5-belief modality.} That is:
	\begin{enumerate}
		\item $\models\K{\texttt{f}}{\texttt{g}}(\phi\rightarrow\phi')\rightarrow
						(\K{\texttt{f}}{\texttt{g}}(\phi)\rightarrow\K{\texttt{f}}{\texttt{g}}(\phi'))$\quad(Kripke's law, K)
		\item $\models\K{\texttt{f}}{\texttt{g}}(\phi)\rightarrow\phi$ \quad(truth law, T)
		\item $\models\K{\texttt{f}}{\texttt{g}}(\phi)\rightarrow\K{\texttt{f}}{\texttt{g}}(\K{\texttt{f}}{\texttt{g}}(\phi))$\quad(positive introspection, \textbf{4})
		\item $\models\neg\K{\texttt{f}}{\texttt{g}}(\phi)\rightarrow\K{\texttt{f}}{\texttt{g}}(\neg\K{\texttt{f}}{\texttt{g}}(\phi))$\quad(negative introspection, \textbf{5})
		\item if $\models\phi$ then $\models\K{\texttt{f}}{\texttt{g}}(\phi)$\quad(necessitation, N)
	\end{enumerate}
\end{proposition}
\begin{myproof}
	By Theorem~\ref{theorem:S5}.1.
	The T-law corresponds to reflexivity.
\end{myproof}
The following theorem summarises our main results.
\begin{theorem}[Doxastic-epistemic modality conditionals]\label{theorem:DoxasticEpistemic}\ 
	\begin{enumerate}
		\item $\boxed{\models\K{\texttt{f}}{\mathtt{id}_{\mathtt{Im}(\texttt{f})}}(\phi)\leftrightarrow\B{\texttt{f}}{\mathtt{id}_{\mathtt{Im}(\texttt{f})}}(\phi)}$\quad(knowledge as unbiased belief)
		\item $\models\B{\texttt{f}}{\mathtt{id}_{\mathtt{Im}(\texttt{f})}}(\phi)\rightarrow\phi$\quad(unbiased belief is true belief)
		\item $\models\phi\rightarrow\K{\mathtt{id}_{\mathtt{S}}}{\mathtt{id}_{\mathtt{S}}}(\phi)$\quad(perfect knowledge)
		\item $\models\K{\mathtt{id}_{\mathtt{S}}}{\mathtt{id}_{\mathtt{S}}}(\phi)\leftrightarrow\B{\mathtt{id}_{\mathtt{S}}}{\mathtt{id}_{\mathtt{S}}}(\phi)$\quad(perfect knowledge as precise belief)
		%\item $\models\B{\texttt{f}}{\texttt{g}}(\phi)\rightarrow\K{\texttt{f}}{\texttt{g}}(\B{\texttt{f}}{\texttt{g}}(\phi))$\quad(belief is known, like in \cite{KnowledgeAsBelief})
		\item $\models\K{\texttt{f}}{\texttt{g}}(\phi)\rightarrow\B{\texttt{f}}{\texttt{g}}(\phi)$\quad(knowledge implies belief, like in \cite{KnowledgeAsBelief})
		\item bias cancellation:
			\begin{enumerate}
		\item $\mathfrak{S}\models\K{\texttt{f}}{\texttt{g}}(\phi)\leftrightarrow\B{\texttt{f}}{\texttt{g}}(\phi)$ if and only if $\iota(\texttt{g})=\mathrm{id}_{\mathrm{Im}(\iota(\texttt{f}))}\,,$ 
				\item $\iota(\texttt{g})=\mathrm{id}_{\mathrm{Im}(\iota(\texttt{f}))}$ if and only if 
						$\iota(\texttt{g})$ is injective;
			\end{enumerate}
		\item $\boxed{\models\neg\B{\mathtt{id}_{\mathtt{S}}}{\texttt{g}}(\phi)\rightarrow\B{\mathtt{id}_{\mathtt{S}}}{\texttt{g}}(\neg\phi)}$\quad(negation-complete belief)
		\item negation-complete knowledge:
			\begin{enumerate}
				\item $\mathfrak{S}\models\neg\K{\mathtt{id}_{\mathtt{S}}}{\texttt{g}}(\phi)\rightarrow\K{\mathtt{id}_{\mathtt{S}}}{\texttt{g}}(\neg\phi)$ if and only if $\iota(\texttt{g})$ is injective,
				\item negation-complete knowledge coincides with 
						perfect knowledge.
			\end{enumerate}
	\end{enumerate}
\end{theorem}
\begin{myproof}
	1 follows from Proposition~\ref{proposition:EDD}.2; 
	2 from 1 and Proposition~\ref{theorem:S5KnowledgeModality}.2; 
	3 and 4 from the left and right equation in Proposition~\ref{proposition:EDD}.3, respectively; 
	(4 also as an instance of 1;)
	%5 from the left inclusion in Proposition~\ref{proposition:EDD}.1; 
	5 from %the right inclusion in 
		Proposition~\ref{proposition:EDD}.1; 
	6.a from Proposition~\ref{proposition:EDD}.2; 
	6.b from the fact that an idempotent function that is also injective must be the identity; 
	7 from the functionality of $\mathrm{D}_{\mathrm{id}_{\states}}^{\iota(\emph{\texttt{g}})}$ for 
		any $\iota$; 
	8.a from the functionality of $\mathrm{E}_{\mathrm{id}_{\states}}^{\iota(\emph{\texttt{g}})}$ for  
		injective $\iota(\emph{\texttt{g}})$; and
	8.b from 8.a, 6.b, and 3.
\end{myproof}

The following proposition establishes formal correspondences to related work.
\begin{proposition}[Related work]\label{proposition:RelatedWork}\ 
\begin{enumerate}
	\item \emph{Epistemic accessibility as state indistinguishability \cite{Epistemic_Logic}:} 
		Let $$\states\ni s ::= \mathtt{0}\mid \alpha_{a}(s)\,,$$ where 
			$\mathtt{0}$ designates a zero data point (e.g., an initial state) and 
			$\alpha_{a}$ an action performed by agent $a\in\mathcal{A}$ for some
						finite set $\mathcal{A}\neq\emptyset$ of agents, and 
		define the visibility $\pi_{a}$ in Definition~\ref{definition:EAasI} on $\states$ as 
			\begin{align*}
				\pi_{a}(\mathtt{0})&\defeq\mathtt{0}\\
				\pi_{a}(\alpha_{b}(s))&\defeq
					\begin{cases}
						\alpha_{b}(\pi_{a}(s)) & \text{if $a=b$, and}\\
						\pi_{a}(s) &\text{otherwise.}
					\end{cases}
			\end{align*}
		Then, 
		$$\boxed{\mathrm{E}_{\pi_{a}}^{\mathrm{id}_{\mathrm{Im}(\pi_{a})}}={\equiv_{a}}\,.}$$ 
		Thus we can reconstruct the standard agent-centric epistemic mo\-dality $\K{a}{}$ 
			\cite{Epistemic_Logic} in our framework with 
			the following simple definition 
				$$\text{$\K{a}{}(\phi) \defeq
					\K{\mathtt{pi}_{a}}{\mathtt{id}_{\mathtt{Im}(\mathtt{pi}_a)}}(\phi)\,$}$$
			for 
				doxastic-epistemic similarity types such that 
					$\mathcal{T}\defeq\{\mathtt{S}\}\cup\{\mathtt{Im}(\mathtt{pi}_a)\mid a\in\mathcal{A}\}$, 
					$\mathcal{G}\defeq\{\mathtt{id}_{\mathtt{Im}(\mathtt{pi}_a)}:\mathtt{Im}(\mathtt{pi}_a)\to \mathtt{Im}(\mathtt{pi}_a)\mid a\in\mathcal{A}\}$, 
					$\mathcal{F}\defeq\{\mathtt{pi}_a : \mathtt{S}\to \mathtt{Im}(\mathtt{pi}_a)\mid a\in\mathcal{A}\}$,
					$\Pi_K\defeq\{(\mathtt{pi}_a,\mathtt{id}_{\mathtt{Im}(\mathtt{pi}_a)})\mid a\in\mathcal{A}\}$, and
					$\Pi_B\defeq\emptyset\,;$ and
			an interpretation function $\iota$ on types and typed function names such that 
					$\iota(\mathtt{Im}(\mathtt{pi}_a))\defeq\mathrm{Im}(\pi_{a})$ and 
					$\iota(\mathtt{pi}_{a})\defeq\pi_{a}\,,$ respectively. 
					
					The resulting instantiation structure is $(\states,\iota)$.
	\item \emph{Epistemic Logic as Dynamic Logic \cite{NonMonotonicKnowledge}:}
			Recall 
				Parikh's embedding $\theta$ \cite{NonMonotonicKnowledge} of 
					Epistemic Logic \cite{Epistemic_Logic} into 
					Propositional Dynamic Logic \cite{Dynamic_Logic} with 
						inverse actions \cite{PDLplusInverse}, by 
							which Parikh established an upper, EXPTIME complexity bound for 
								Epistemic Logic (also with common knowledge): 
						\begin{eqnarray*}
							\theta(P)&\defeq&P\\
							\theta(\neg\phi)&\defeq&\neg\theta(\phi)\\
							\theta(\phi\wedge\phi')&\defeq&\theta(\phi)\wedge\theta(\phi')\\
							\theta(\K{a}{}(\phi))&\defeq&[(\alpha_{a}\cup(\alpha_{a})^{-1})^{*}]\theta(\phi)\,,
						\end{eqnarray*}
						where $[(\alpha_{a}\cup(\alpha_{a})^{-1})^{*}]$ is  
							the dynamic necessity modality with
								the program parameter $(\alpha_{a}\cup(\alpha_{a})^{-1})^{*}$ for  
									$\alpha_{a}$ as before.
					Further,
						let $\alpha$ denote actions, and $A$ and $A'$ action terms  
							such as $(\alpha_{a}\cup(\alpha_{a})^{-1})^{*}$, and let 
							$$\begin{array}{@{}rcl@{\qquad}rcl@{}}
								\states\ni s &::=& \mathtt{0}\mid \alpha_{a}(s) &
								\mathrm{R}_{A^{0}}&\defeq&\mathrm{id}_{\states}\\[\jot]
								\mathrm{R}_{\alpha}&\defeq&\{(s,\alpha(s))\mid s\in\states\} &
								\mathrm{R}_{A^{1}}&\defeq&\mathrm{R}_{A}\\[\jot]
								\mathrm{R}_{A^{-1}}&\defeq&(\mathrm{R}_{A})^{-1} &
								\mathrm{R}_{A^{n+1}}&\defeq&\mathrm{R}_{A^{n}}\circ \mathrm{R}_{A^{1}}\\[\jot]
								\mathrm{R}_{A\cup A'}&\defeq&\mathrm{R}_{A}\cup \mathrm{R}_{A'} &
								\mathrm{R}_{A^{*}}&\defeq&\bigcup_{n\in\mathbb{N}}\mathrm{R}_{A^{n}}\,.
							\end{array}$$
					Then, 
						$$\boxed{\mathrm{id}_{\states}\cup\mathrm{E}_{\mathrm{id}_{\states}}^{\mathrm{R}_{\alpha_{a}}}=\mathrm{R}_{(\alpha_{a}\cup(\alpha_{a})^{-1})^{*}}\,,}$$
						where $\mathrm{R}_{(\alpha_{a}\cup(\alpha_{a})^{-1})^{*}}$ is of course an equivalence relation.
						
					Notice that $\mathrm{R}_{\alpha_{a}}$ is not idempotent, and so
						$(\mathrm{id}_{S},\mathrm{R}_{\alpha_{a}})$ is not a 
							doxastic-epistemic function pair, but fortunately 
					thanks to Theorem~\ref{theorem:S5}.2, there is a constructible 
						$\pi_{a}:\states\to\states$ such that 
							\begin{enumerate}
								\item $\mathrm{E}_{\pi_{a}}^{\mathrm{id}_{\states}}=
							\mathrm{id}_{\states}\cup\mathrm{E}_{\mathrm{id}_{\states}}^{\mathrm{R}_{\alpha_{a}}},$ and 
								\item $(\pi_{a},\mathrm{id}_{\states})$ is a doxastic-epistemic function pair.
							\end{enumerate}
	\item Interactive Provability as Explicit Belief \cite{LDiiP}: From \cite{LDiiP}, recall 
			the definition of the (idempotent) operator 
					$\sigma_{a}^{M}:\states\rightarrow\states$ defined on 
					$\states\ni s ::= \mathtt{0} \mid \mathtt{succ}_{a}^{M}(s)$ such that 
							$$\sigma_{a}^{M}(s)\defeq
							\begin{cases}
								s & \text{if $M\in\mathrm{cl}_{a}^{s}(\emptyset)$, and}\\
								\mathtt{succ}_{a}^{M}(s) & \text{otherwise (oracle input),}
							\end{cases}$$
					where 
						$M$ designates a proof term,
						$\mathtt{succ}_{a}^{M}$ a state constructor, and
						$\mathrm{cl}_{a}^{s}$ a closure operator such that 
							$M\in\mathrm{cl}_{a}^{\mathtt{succ}_{a}^{M}(s)}(\emptyset)$.
					(Here, the exact nature of $M$ and $\mathrm{cl}_{a}^{s}$ is unimportant.)
					Then, 
						$$\boxed{\mathrm{D}_{\mathrm{id}_{\states}}^{\sigma_{a}^{M}}={_{M}\mathrm{R}_{a}}\,,}$$
					where ${_{M}\mathrm{R}_{a}}$ is 
						the accessibility relation for 
							the negation-complete proof modality axiomatised in \cite{LDiiP}, obtainable  
								directly as $\mathrm{D}_{\mathrm{id}_{\states}}^{\sigma_{a}^{M}}$.
\end{enumerate}
\end{proposition}
\begin{myproof}
	For 1, inspect 
		Definition~\ref{definition:FunctionParametricEpistemicAccessibility} and 
		\ref{definition:EAasI}.
	
	For 2, consider:
	\begin{align*}
			\mathrm{id}_{\states}\cup\mathrm{E}_{\mathrm{id}_{\states}}^{\mathrm{R}_{\alpha_{a}}}&=
				\mathrm{id}_{\states}\cup(\mathrm{D}_{\mathrm{id}_{\states}}^{\mathrm{R}_{\alpha_{a}}}\cup
				 (\mathrm{D}_{\mathrm{id}_{\states}}^{\mathrm{R}_{\alpha_{a}}})^{-1})^{+}\\
			&=(\mathrm{D}_{\mathrm{id}_{\states}}^{\mathrm{R}_{\alpha_{a}}}\cup
				 (\mathrm{D}_{\mathrm{id}_{\states}}^{\mathrm{R}_{\alpha_{a}}})^{-1})^{*}\\
			&=(
				\mathrm{R}_{\alpha_{a}}\cup
				 (\mathrm{R}_{\alpha_{a}})^{-1})^{*}\qquad\text{($\mathrm{R}_{\alpha_{a}}$ is functional)}\\
			&=(
				\mathrm{R}_{\alpha_{a}}\cup
				 \mathrm{R}_{(\alpha_{a})^{-1}})^{*}\\
			&=(
				\mathrm{R}_{\alpha_{a}\cup
							 (\alpha_{a})^{-1}})^{*}\\
			&=\mathrm{R}_{(\alpha_{a}\cup(\alpha_{a})^{-1})^{*}}\,.
	\end{align*}
			
	For 3, inspect Definition~\ref{definition:FunctionParametricDoxasticAccessibility} and \cite{LDiiP}.
\end{myproof}

\section{Conclusion}
We conclude by	
	mentioning that 
		from $\mathrm{D}_{f}^{g}$ and $\mathrm{E}_{f}^{g}$, 
			we can further construct accessibility relations for modalities of 
					\emph{common and distributed belief and knowledge} in a standard way \cite{Epistemic_Logic,MultiAgents} by taking 
						unions and transitive closures of $\mathrm{D}_{f}^{g}$-relations for 
							common belief, 
						unions and reflexive-transitive closures of $\mathrm{E}_{f}^{g}$-relations for 
							common knowledge, and
						intersections of $\mathrm{D}_{f}^{g}$- and $\mathrm{E}_{f}^{g}$-relations for 
							distributed belief and knowledge, respectively.

For example,
	accessibility relations
		$\mathrm{DD}_{\mathcal{C}}^{g}$ for distributed belief, 
		$\mathrm{CD}_{\mathcal{C}}^{g}$ for common belief, 		
		$\mathrm{DE}_{\mathcal{C}}^{g}$ for distributed knowledge, and 
		$\mathrm{CE}_{\mathcal{C}}^{g}$ for common knowledge with respect to 
			doxastic-epistemic function pairs $(\pi_{a},g)$ and 
			a community $\mathcal{C}$ of agents $a$ can be constructed:
				$$\begin{array}{rcl@{\qquad}rcl}
					\mathrm{DD}_{\mathcal{C}}^{g}&\defeq&
						\bigcap_{a\in\mathcal{C}}\mathrm{D}_{\pi_{a}}^{g} &
					\mathrm{CD}_{\mathcal{C}}^{g}&\defeq&
						(\bigcup_{a\in\mathcal{C}}\mathrm{D}_{\pi_{a}}^{g})^{+}\\[2\jot]
					\mathrm{DE}_{\mathcal{C}}^{g}&\defeq&
						\bigcap_{a\in\mathcal{C}}\mathrm{E}_{\pi_{a}}^{g} &
					\mathrm{CE}_{\mathcal{C}}^{g}&\defeq&
						(\bigcup_{a\in\mathcal{C}}\mathrm{E}_{\pi_{a}}^{g})^{*}\\[\jot]
					\mathrm{DE}_{\mathcal{C}}&\defeq&
						\mathrm{DE}_{\mathcal{C}}^{\mathrm{id}_{\states}} &
					\mathrm{CE}_{\mathcal{C}}&\defeq&
						\mathrm{CE}_{\mathcal{C}}^{\mathrm{id}_{\states}}\,.
				\end{array}$$

\bibliographystyle{alpha}

\end{document}